\documentclass[11pt]{article}
\usepackage{mathrsfs}
\usepackage{amssymb}
\usepackage{amsfonts}
\usepackage{color}
\usepackage{amsthm}
\usepackage{amsmath}
\usepackage{ascmac}
\usepackage[top=35truemm, bottom=30truemm, left=30truemm, right=30truemm]{geometry}
\usepackage{fancybox}
\usepackage{latexsym}

\newtheorem{thm}{Theorem}[section]
\newtheorem{lem}[thm]{Lemma}

\title{Stochastic process emerged from lattice fermion systems by repeated measurements and large-time limit}
\author{Kazuki Yamaga\thanks{Department of Nuclear Engineering, Kyoto University, yamaga.kazuki.62a@st.kyoto-u.ac.jp}}
\date{}

\begin{document}
\maketitle

\begin{abstract}
It is known that in quantum theory, measurements may suppress Hamiltonian dynamics of a system. A famous example is the `Quantum Zeno Effect'. This is the phenomena that if one performs the measurements $M$ times asking whether the system is in the same state as the one at the initial time until the fixed measurement time $t$, then survival probability tends to $1$ by taking the limit $M\to\infty$. This is the case for fixed measurement time $t$. It is known that if one takes measurement time infinite at appropriate scaling, `Quantum Zeno Effect' does not occur and the effect of Hamiltonian dynamics emerges \cite{doi:10.1063/1.4978851}. In the present paper, we consider the long time repeated measurements and the dynamics of quantum many body systems in the scaling where the effect of measurements and dynamics are balanced. We show that the stochastic process, called symmetric simple exclusion process (SSEP), is obtained from the repeated and long time measurements of configuration of particles in finite lattice fermion systems. The emerging stochastic process is independent of potential and interaction of the underlying Hamiltonian of the system.
\end{abstract}

\section{Introduction}
It is known that in quantum theory, measurements suppress Hamiltonian dynamics of a system. A famous example is the `Quantum Zeno Effect' \cite{doi:10.1063/1.523304}. This phenomena states that for fixed finite time $t$ if one performs repeated measurements in small interval (taking limit to $0$), then the dynamics of the system freezes. More precisely, suppose the system is initially in the (pure) state $\psi_0$ and evolves by the Hamiltonian $H$. For fixed time $t$, one repeats the 2-outcome measurements in the interval $t/M$ asking if the system is in the state `$\psi_0$' or not, which is described by the PVM $\{|\psi_0\rangle\langle\psi_0|,1-|\psi_0\rangle\langle\psi_0|\}$: repeat the following processes until $t$,
\[ \mathrm{dynamics\ :\ }\psi\mapsto e^{-i\frac{t}{M} H}\psi \]
\[ \mathrm{measurement\ :\ }''\psi_0''\ \mathrm{or\ not?}. \]
The probability getting the outcome `$\psi_0$' in all the measurements tends to $1$ in the limit $M\to\infty$ (small interval limit). That is, the dynamics of the system is frozen by the continuous measurement. This is the `Quantum Zeno Effect'.

However, this is the case of finite measurement time $t$. It is known that if one takes measurement time infinite at appropriate scaling, `Quantum Zeno Effect' does not occur and the effect of the Hamiltonian dynamics emerges \cite{doi:10.1063/1.4978851}. In this paper, we will consider the case that suppression by repeated measurements and Hamiltonian dynamics are balanced. In the setting of 'Quantum Zeno Effect', one usually considers the 2-outcome measurement $\{|\psi_0\rangle\langle\psi_0|,1-|\psi_0\rangle\langle\psi_0|\}$ asking whether the system is in the state `$\psi_0$' or not. If one deals with more complicated outcome space such as configuration of particles, physically meaningful stochastic processes may be obtained. 

Here, we consider the measurement of configuration of particles in finite lattice fermion systems. We deal with the Hamiltonian consisting of hopping term, potential term and 2-body interaction term. For fixed $\tau$, one repeats the measurement of configuration until $\tau M$ in the interval $1/M$. That is, the number of measurements is $[\tau M^2]$, the maximum integer which does not exceed $\tau M^2$. Taking the limit $M\to\infty$, we obtain a stochastic process with a new time $\tau$. This process corresponds to the symmetric simple exclusive process (SSEP) and is independent of potential and interaction terms of the Hamiltonian. It is known that the diffusion equation is obtained from the SSEP by the appropriate scaling limit called hydrodynamic limit \cite{kipnis1998scaling, guo1988}. By the diffusion equation, the diffusion of particles is proportional to $\sqrt{\tau}$. If the measurements are not performed, generally the transport property of quantum many body systems should be influenced by the potential \cite{abdul2017localization,damanik2015quantum, damanik2015transport,kachkovskiy2016transport}. For example, if the potential is random the system shows the localization (Anderson localization) \cite{germinet1998dynamical,carmona1987,PhysRev.109.1492} and if the potential is periodic the system shows the ballistic transport \cite{bruneau2016conductance}. Thus, our result suggests that the effect of repeated measurements promotes the diffusion for the random potentials and suppresses the transport of particles for the periodic potentials. 

\section{Lattice fermion system on the circle}
In this section, we recall the description of many body fermion systems on lattice. Here consider the one-dimensional finite lattice $\{1,2,\cdots,N\}$. Many body fermion systems on this lattice is described by creation and annihilation operators $a_n^*,a_n\ (n=1,2,\cdots,N)$ satisfying the following canonical anti-commutation relations: 
\[ \{a_n,a_m\}=0,\ \{a_n^*,a_m\}=\delta_{nm}1. \]
These operators act on the fermion Fock space ($2^N$ dimension) consisting of one-particle Hilbert space $\mathbb{C}^N$. In this paper, we consider the following form of Hamiltonian
\[ H=\sum_{n=1}^N\left[-\frac{1}{2}(a_n^*a_{n+1}+a_{n+1}^*a_n)+v(n)a_n^*a_n+\lambda a^*_{n}a_{n}a^*_{n+1}a_{n+1}\right], \]
where $v\colon\{1,\cdots,N\}\to\mathbb{R}$ is a real valued function called potential and $\lambda\in\mathbb{R}$ is a parameter representing the strengthen of interaction. $-\frac{1}{2}(a_n^*a_{n+1}+a_{n+1}^*a_n),\ v(n)a_n^*a_n$ and $\lambda a^*_{n}a_{n}a^*_{n+1}a_{n+1}$ represent hopping, potential and interaction terms respectively. We consider the periodic boundary condition and identify $a_{N+1}$ as $a_1$.

Put $A_n^0=a_na_n^*,\ A_n^1=a_n^*a_n$ ($n=1,\cdots,N$), then from canonical anti-commutation relations it turns out that they are projections commuting each other. For a configuration of particles $x=(x_1,x_2,\cdots,x_N)\in\{0,1\}^N$ ($0,1$ correspond to the absence and the existence of a particle respectively, and $x_n$ represents whether a particle is in the site $n$ or not), put
\[ P_{(x_1,\cdots,x_N)}=\prod_{n=1}^NA_n^{x_n}. \]
Then they are projections and satisfy
\[ \sum_{x\in\{0,1\}^N}P_x=I. \]
That is $\{P_x\}_{x\in\{0,1\}^N}$ is a PVM (projection-valued measure) representing the measurement of configuration of particles. Since $P_x\neq0$ and the number of outcomes is equal to the dimension of the Hilbert space they are 1-rank projections. In this paper, we consider only projection measurement, that is, if one performs the measurement $\{P_x\}_{x\in\{0,1\}^N}$ on the system in the state $\rho$ and obtain the outcome $x$, then the state after the measurement is 
\[ \frac{P_x\rho P_x}{\mathrm{Tr}\rho P_x} .\]

\section{SSEP from repeated measurement}
First, let us consider the Hamiltonian without potential and interaction terms:
\[ H=-\frac{1}{2}\sum_{n=1}^N(a_n^*a_{n+1}+a_{n+1}^*a_n), \]
$a_{N+1}=a_1$ (periodic boundary condition). The system evolves by this Hamiltonian.

Suppose that we repeat the measurements of configuration on the system initially in the state $\rho$ (we identify the density operator $\rho$ and the expectation value functional $A\mapsto\mathrm{Tr}\rho A$, and use the same symbol) until $T$ with interval $t$. Put $L\in\mathbb{N},\ s\in\mathbb{R}$ which satisfies $T=Lt+s\ (L\in\mathbb{N},0\le s<t)$. Then, the probability $p_T^t(x)$ of getting the outcome $x$ by the configuration measurement at time $T$ is 
\[
 p_T^t(x)=\sum_{x_L\in\{0,1\}^N}\cdots\sum_{x_1\in\{0,1\}^N}\omega_{x_L}(e^{isH}P_xe^{-isH})\omega_{x_{L-1}}(e^{itH}P_{x_L}e^{-itH})\cdots\omega_{x_1}(e^{itH}P_{x_2}e^{-itH})\rho(e^{itH}P_{x_1}e^{-itH}), 
\]
where $\omega_x$ is the state which has the density operator $P_x$. Put
\[ p_0^t(x)=\rho(e^{itH}P_xe^{-itH}) \]
and define a $2^N\times 2^N$ matrix $U_t$ with $(x,y)$-entry
\[ \omega_y(e^{itH}P_xe^{-itH}). \]
Since
\[ \sum_{x\in\{0,1\}^N}\omega_y(e^{itH}P_xe^{-itH})=\sum_{y\in\{0,1\}^N}\omega_y(e^{itH}P_xe^{-itH})=1, \] 
$U_t$ is a doubly stochastic matrix. With $U_t$, the probability distribution $p_T^t$ is expressed as 
\[ p_T^t=U_s(U_t)^{L-1}p_0^t. \]
Let us make the measurement interval $t$ small and the measurement time $T$ large. Fix $\tau>0$ and let $M$ be a positive integer. Put $t=\frac{1}{M}$ and $T=\tau M$ and take the limit $M\to\infty$.

Here, we state our main result as a theorem.

\begin{thm}
$p_{\tau M}^{\frac{1}{M}}$ converges to a probability distribution $q_\tau$ on $\{0,1\}^N$ by the limit $M\to\infty$. This distribution corresponds to that of symmetric simple exclusion process (SSEP) initially in the distribution $\{\rho(P_x)\}_{x\in\{0,1\}^N}$.
\end{thm}
In the following, we provide the proof of this result step by step.

For $a\in\mathbb{R}$, let us denote $[a]$ the maximum integer which does not exceed $a$. Then $L=[\tau M^2]$. $U_s$ ($0\le s<\frac{1}{M}$) in the right hand side of  
\[ p_{\tau M}^{\frac{1}{M}}=U_s\left(U_{\frac{1}{M}}\right)^{[\tau M^2]-1}p_0^{\frac{1}{M}} \]
converges to the identity operator and $p_0^{\frac{1}{M}}(x)$ to $p_0^0(x)=\rho(P_x)$ as the limit $M\to\infty$. Let us focus on the factor
\[ \left(U_{\frac{1}{M}}\right)^{[\tau M^2]-1}. \]
Expanding $\omega_y(e^{itH}P_xe^{-itH})$ in terms of $t$, we have
\[ \omega_y(e^{itH}P_xe^{-itH})=\delta_{xy}+it\omega_y([H,P_x])-\frac{t^2}{2}\omega_y([H,[H,P_x]])+O(t^3) , \]
where $[A,B]=AB-BA$. Since $[P_x,P_y]=0$, the second term is $0$:
\[ \omega_y([H,P_x])=\mathrm{Tr}(P_xP_yH-P_yP_xH)=0. \]
Defining a $2^N\times 2^N$ matrix $X$ as
\[ (X)_{xy}=-\frac{1}{2}\omega_y([H,[H,P_x]]) ,\]
then we get
\begin{equation}
\lim_{M\to\infty}\left(U_{\frac{1}{M}}\right)^{[\tau M^2]-1}=e^{\tau X}. \end{equation}
In order to prove this fact, we prepare a lemma. 

\begin{lem}
Let $V$ be a Banach space and $X$ be a bounded operator on $V$. And suppose $\{Y_K\}_{K\in\mathbb{N}}$ is a sequence of bounded operators on $V$ such that $K\|Y_K\|\to0\ (K\to\infty)$ $(\|Y_K\|$ is a operator norm of $Y_K)$. Then we obtain
\[ \lim_{K\to\infty}\left(1+\frac{X}{K}+Y_K\right)^K=e^X, \]
in the operator norm.
\end{lem}
\begin{proof}
The proof consists of two parts. First, we show the relation which is well-known for the case that $X$ is a number, 
\[ \lim_{K\to\infty}\left(1+\frac{X}{K}\right)^K=e^X .\]
Recall that
\[ \left(1+\frac{X}{K}\right)^K=\sum_{n=0}^K
\left(
\begin{array}{c}
K \\
n
\end{array}
\right)\frac{X^n}{K^n}=\sum_{n=0}^K\frac{X^n}{n!}\frac{K(K-1)\cdots(K-n+1)}{K^n} ,\]
\[ e^X=\sum_{n=0}^\infty\frac{X^n}{n!}. \]
For $\epsilon>0$, there exists a positive integer $K_0\in\mathbb{N}$ such that $\displaystyle\sum_{n=K_0+1}^\infty\frac{\|X\|^n}{n!}<\frac{\epsilon}{3}$. For $K>K_0$ by the inequality
\begin{eqnarray*}
\left\|\left(1+\frac{X}{K}\right)^K-e^X\right\|&\le&\left\|\sum_{n=0}^{K_0}\frac{X^n}{n!}\frac{K(K-1)\cdots(K-n+1)}{K^n}-\sum_{n=0}^{K_0}\frac{X^n}{n!} \right\| \\
&&+\left\|\sum_{n=K_0+1}^K\frac{X^n}{n!}\frac{K(K-1)\cdots(K-n+1)}{K^n}\right\|+\left\|\sum_{n=K_0+1}^\infty\frac{X^n}{n!}\right\| \\
&\le&\left\|\sum_{n=0}^{K_0}\frac{X^n}{n!}\frac{K(K-1)\cdots(K-n+1)}{K^n}-\sum_{n=0}^{K_0}\frac{X^n}{n!}\right\|+2\sum_{n=K_0+1}^\infty\frac{\|X\|^n}{n!},
\end{eqnarray*}
for sufficiently large $K$, the first term of the right hand side is smaller than $\frac{\epsilon}{3}$ and the right hand side is bounded from above by $\epsilon$.

Next, we will show that
\[ \lim_{K\to\infty}\left[\left(1+\frac{X}{K}\right)^K-\left(1+\frac{X}{K}+Y_K\right)^K\right]=0. \]
By expanding $\left(1+\frac{X}{K}+Y_K\right)^K$, we get
\[ \left\|\left(1+\frac{X}{K}\right)^K-\left(1+\frac{X}{K}+Y_K\right)^K \right\|\le\sum_{n=1}^K\left\|1+\frac{X}{K}\right\|^{K-n}\|Y_K\|^n\frac{K(K-1)\cdots(K-n+1)}{n!}. \]
The right hand side is bounded above by
\[ \left(1+\frac{\|X\|}{K}\right)^K\sum_{n=1}^K\frac{(K\|Y_K\|)^n}{n!}\le\left(1+\frac{\|X\|}{K}\right)^K(e^{K\|Y_K\|}-1). \]
Since $K\|Y_K\|\to0$ and $\left(1+\frac{\|X\|}{K}\right)^K\to e^{\|X\|}$ as $K\to\infty$, the right hand side of the inequality tends to $0$ as $K\to\infty$.
\end{proof}

\begin{proof}[The proof of equation (1)]
The proof consists of the following two steps:
\begin{itemize}
\item Show $\displaystyle\lim_{M\to\infty}\left(U_\frac{1}{M}\right)^{\tau M^2}=e^{\tau X}$ by using Lemma 3.2 for the case $K=M^2$.
\item Show $\displaystyle\lim_{M\to\infty}\left\|\left(U_{\frac{1}{M}}\right)^{[\tau M^2]-1}-\left(U_\frac{1}{M}\right)^{\tau M^2}\right\|=0$.
\end{itemize}
Since
\[ U_\frac{1}{M}=I+\frac{X}{M^2}+ \cdots,\]
in order to apply Lemma 3.2 for the case $K=M^2$, we have to show 
\[ M^2\left\|U_\frac{1}{M}-I-\frac{X}{M^2}\right\|\to0\ (M\to\infty). \]

Putting $\delta(A)=[H,A]$, then by the inequality $\|\delta(A)\|\le2\|H\|\|A\|$ we have 
\begin{eqnarray*}
\left|\omega_y(e^{itH}P_xe^{-itH})-\delta_{xy}+\frac{t^2}{2}\omega_y([H,[H,P_x]])\right|&=&\left|\sum_{n=3}^\infty\omega_y\left(\frac{t^n\delta^n}{n!}(P_x)\right)\right| \\
&\le&|t|^3\sum_{n=3}^\infty\frac{(2\|H\|)^n}{n!} \\
&\le&|t|^3e^{2\|H\|}
\end{eqnarray*}
for $|t|<1$. Thus, 
\[ M^2\left\|U_{\frac{1}{M}}-I-\frac{X}{M^2}\right\|\le \frac{2^N}{M}e^{2\|H\|}\to0\ (M\to\infty), \]
and by Lemma 3.2 we obtain
\[ \lim_{M\to\infty}\left(U_\frac{1}{M}\right)^{\tau M^2}=\left(\lim_{M\to\infty}\left(U_\frac{1}{M}\right)^{M^2}\right)^\tau=e^{\tau X}. \]
Next, we estimate the difference between $\left(U_\frac{1}{M}\right)^{[\tau M^2]-1}$ and $\left(U_\frac{1}{M}\right)^{\tau M^2}$. Denote $Y_M = U_\frac{1}{M}-I-\frac{X}{M^2}$, then
\begin{eqnarray*}
&&\left\|\left(I+\frac{X}{M^2}+Y_M\right)^{[\tau M^2]-1}-\left(I+\frac{X}{M^2}+Y_M\right)^{\tau M^2}\right\| \\
&\le&\left(1+\frac{\|X\|}{M^2}+\|Y_M\|\right)^{[\tau M^2]-1}\left\|I-\left(I+\frac{X}{M^2}+Y_M\right)^{1+\tau M^2-[\tau M^2]}\right\| .
\end{eqnarray*}
The first factor of the right hand side tends to $e^{\tau\|X\|}$ as $M\to\infty$. Let us consider the second factor.
\begin{eqnarray*}
\left\|I-\left(I+\frac{X}{M^2}+Y_M\right)^{1+\tau M^2-[\tau M^2]}\right\|&\le&\left\|\frac{X}{M^2}+Y_M\right\| \\
&&+ \left\|I+\frac{X}{M^2}+Y_M\right\|\left\|I-\left(I+\frac{X}{M^2}+Y_M\right)^{\tau M^2-[\tau M^2]}\right\|.
\end{eqnarray*}
Setting $A_M=\frac{1}{M^2}X+Y_M$ and $a=\tau M^2-[\tau M^2]$, then $0\le a<1$. And since $\|A_M\|\to0$ as $M\to\infty$, $\|A_M\|<1$ for large $M$. By 
\[ (I+A_M)^a=\sum_{n=0}^\infty\left(
\begin{array}{c}
a \\
n
\end{array}\right)
A_M^n, \]
we have 
\[ \|I-(I+A_M)^a\|\le\sum_{n=1}^\infty\left|\left(
\begin{array}{c}
a \\
n
\end{array}\right)\right|
\|A_M\|^n,\]
where 
\[ \left(
\begin{array}{c}
a \\
n
\end{array}\right)
=\frac{a(a-1)\cdots(a-n+1)}{n!} .\]
Since $0\le a<1$, $(-1)^n\left(
\begin{array}{c}
a \\
n
\end{array}\right)\le0$. Thus we obtain
\[ \sum_{n=1}^\infty\left|\left(
\begin{array}{c}
a \\
n
\end{array}\right)\right|
\|A_M\|^n=-\sum_{n=0}^\infty\left(
\begin{array}{c}
a \\
n
\end{array}\right)(-\|A_M\|)^n+1=1-(1-\|A_M\|)^a\le\|A_M\| \]
and this goes to $0$ as $M\to\infty$. Combining the above discussions, we get the conclusion
\[ \lim_{M\to\infty}\left(U_{\frac{1}{M}}\right)^{[\tau M^2]-1}=\lim_{M\to\infty}\left(U_\frac{1}{M}\right)^{\tau M^2}=e^{\tau X}. \]
\end{proof}

Using the above discussions, we obtain the limit
\[ q_\tau(x)\equiv \left(e^{X\tau}p_0^0\right)(x)=\lim_{M\to\infty}p_{\tau M}^{\frac{1}{M}}(x) \]
and it turns out that this is the solution of the following equations
\[ \frac{d}{d\tau}q_\tau=Xq_\tau, \]
\[ q_0(x)=\rho(P_x) .\]
$q_\tau$ represents the distribution of the configuration after the large time repeated measurement. The next question is from what stochastic process is this distribution obtained? Let us evaluate the detail of $X$. Recall that the $(x,y)$-entry of $X$ is $-\frac{1}{2}\omega_y([H,[H,P_x]])$. First, in order to obtain $[H,P_x]$, let us calculate $[a_n^*a_{n+1},P_x]$ and $[a_{n+1}^*a_n,P_x]$.
\begin{eqnarray*}
[a_n^*a_{n+1},P_x]&=&\delta_{x_n0}\delta_{x_{n+1}1}a_n^*a_{n+1}\prod_{m\neq n,n+1}^ NA_m^{x_m}\\
&&-\delta_{x_n1}\delta_{x_{n+1}0}\prod_{m\neq n,n+1}^ NA_m^{x_m}a_n^*a_{n+1} \\
&=&(\delta_{x_n0}\delta_{x_{n+1}1}-\delta_{x_n1}\delta_{x_{n+1}0})a_n^*a_{n+1}P_x^{n,n+1},
\end{eqnarray*}
where $P_x^{n,n+1}=\prod_{m\neq n,n+1}^ NA_m^{x_m}$. Similarly we have 
\[ [a_{n+1}^*a_n,P_x]=(\delta_{x_n1}\delta_{x_{n+1}0}-\delta_{x_n0}\delta_{x_{n+1}1})a_{n+1}^*a_nP_x^{n,n+1}. \]
Combining the above equations, we obtain
\[ [H,P_x]=-\frac{1}{2}\sum_{n=1}^N\left[ (\delta_{x_n0}\delta_{x_{n+1}1}-\delta_{x_n1}\delta_{x_{n+1}0})a_n^*a_{n+1}P_x^{n,n+1}+(\delta_{x_n1}\delta_{x_{n+1}0}-\delta_{x_n0}\delta_{x_{n+1}1})a_{n+1}^*a_nP_x^{n,n+1}\right] .\]
By the simple calculation, we have
\[ \omega_y([a_n^*a_{n+1},[a_m^*a_{m+1},P_x]])=0 ,\]
\[ \omega_y([a_n^*a_{n+1},[a_{m+1}^*a_m,P_x]])=\delta_{nm}(A_n^1A_{n+1}^0P_x^{n,n+1}-A_n^0A_{n+1}^1P_x^{n,n+1}) ,\]
\[ \omega_y([a_{n+1}^*a_n,[a_m^*a_{m+1},P_x]])=\delta_{nm}(A_n^0A_{n+1}^1P_x^{n,n+1}-A_n^1A_{n+1}^0P_x^{n,n+1} ),\]
\[ \omega_y([a_{n+1}^*a_n,[a_{m+1}^*a_m,P_x]])=0. \]
Thus, finally we get
\begin{eqnarray*}
-\frac{1}{2}\omega_y([H,[H,P_x]])&=&-\frac{1}{8}\sum_{n=1}^N[ (\delta_{x_n0}\delta_{x_{n+1}1}-\delta_{x_n1}\delta_{x_{n+1}0})\omega_y(A_n^0A_{n+1}^1P_x^{n,n+1}-A_n^1A_{n+1}^0P_x^{n,n+1})  \\
&&+(\delta_{x_n1}\delta_{x_{n+1}0}-\delta_{x_n0}\delta_{x_{n+1}1})\omega_y(A_n^1A_{n+1}^0P_x^{n,n+1}-A_n^0A_{n+1}^1P_x^{n,n+1}) ] \\
&=&-\frac{1}{4}\sum_{n=1}^N[(\delta_{x_n0}\delta_{x_{n+1}1}-\delta_{x_n1}\delta_{x_{n+1}0})\omega_y(A_n^0A_{n+1}^1P_x^{n,n+1}) \\
&& +(\delta_{x_n1}\delta_{x_{n+1}0}-\delta_{x_n0}\delta_{x_{n+1}1})\omega_y(A_n^1A_{n+1}^0P_x^{n,n+1})].
\end{eqnarray*}

When one considers the time evolution of the observables instead of distribution (Heisenberg picture), the generator is the transpose $X^T$ of $X$. The action of $X^T$ to the observable $f\colon\{0,1\}^N\to\mathbb{R}$ is 
\begin{eqnarray*}
 (X^Tf)(y)&=&-\frac{1}{2}\sum_{x\in\{0,1\}^N}\omega_y([H,[H,P_x]])f(x) \\
&=&-\frac{1}{4}\sum_{n=1}^N\left(1_{\{y_n=1,y_{n+1}=0\}}(f(y^{n\leftrightarrow n+1})-f(y))+1_{\{y_n=0,y_{n+1}=1\}}(f(y^{n\leftrightarrow n+1})-f(y)) \right),
\end{eqnarray*}
where for $y\in\{0,1\}^N$, $y^{n\leftrightarrow n+1}$ represents the configuration that exchanges the values at $n$ and $n+1$. $1_{\{y_n=1,y_{n+1}=0\}}$ is $1$ if the condition in $\{\}$ is satisfied and $0$ otherwise.

The stochastic process with such a generator is called symmetric simple exclusion process (SSEP). Theorem 3.1 is proved for the case that Hamiltonian does not include potential and interaction terms.

Before considering the case with potential and interaction terms, we would like to mention the importance of SSEP in (non-equilibrium) statistical physics. SSEP is a special case of a more general model, asymmetric simple exclusion process (ASEP) \cite{spitzer1970interaction, liggett1985interacting}, which is a solvable model of interacting particle systems. Its dynamics and stationary state are well investigated. Moreover, as mentioned in the introduction, it is known that the diffusion equation, $\frac{\partial}{\partial t}\rho(t,x)=D\frac{\partial^2}{\partial x^2}\rho(t,x)$, is obtained from SSEP by the hydrodynamic limit \cite{kipnis1998scaling, guo1988}.

In the following let us consider the Hamiltonian including the potential $\sum_{n=1}^Nv(n)a_n^*a_n$ and the interaction $\lambda\sum_{n=1}^N a_n^*a_na_{n+1}^*a_{n+1}$, and complete the proof of Theorem 3.1. Since these terms commute with $P_x$, they do not change $[H,P_x]$. Let us consider the contribution to $[H,[H,P_x]]$. Calculating the terms which do not become $0$, from the potential term we have 
\[ [a^*_na_n,a_n^*a_{n+1}P_x^{n,n+1}]=a_n^*a_{n+1}P_x^{n,n+1} ,\]
\[ [a^*_{n+1}a_{n+1},a_n^*a_{n+1}P_x^{n,n+1}]=-a_n^*a_{n+1}P_x^{n,n+1}, \]
\[ [a^*_na_n,a_{n+1}^*a_nP_x^{n,n+1}]=-a_{n+1}^*a_nP_x^{n,n+1}, \]
\[ [a^*_{n+1}a_{n+1},a_{n+1}^*a_nP_x^{n,n+1}]=a_{n+1}^*a_nP_x^{n,n+1}. \]
And from the interaction term, we have
\[ [a^*_{n+1}a_{n+1}a^*_{n+2}a_{n+2},a_n^*a_{n+1}P_x^{n,n+1}]=-A_{n+2}^1a_n^*a_{n+1}P_x^{n,n+1} ,\]
\[ [a^*_{n+1}a_{n+1}a^*_{n+2}a_{n+2},a_{n+1}^*a_nP_x^{n,n+1}]=A_{n+2}^1,a_{n+1}^*a_nP_x^{n,n+1}, \]
\[ [a^*_{n-1}a_{n-1}a^*_na_n,a_n^*a_{n+1}P_x^{n,n+1}]=A_{n-1}^1a_n^*a_{n+1}P_x^{n,n+1}, \]
\[ [a^*_{n-1}a_{n-1}a^*_na_n,a_{n+1}^*a_nP_x^{n,n+1}]=-A_{n-1}^1a_{n+1}^*a_nP_x^{n,n+1}. \]
The expectation values of these terms with respect to the state $\omega_y$ are $0$. This is due to the relation $\omega_y(A)=\mathrm{Tr}P_yAP_y$ and the fact that they are $0$ if multiplied $P_y$ from both side. Therefore even if one considers the Hamiltonian including the potential and the interaction
\[ H=\sum_{n=1}^N\left[-\frac{1}{2}(a_n^*a_{n+1}+a_{n+1}^*a_n)+v(n)a_n^*a_n+\lambda a^*_{n}a_{n}a^*_{n+1}a_{n+1}\right], \]
the stochastic process obtained by the large time repeated measurements of configuration is not changed. This completes the proof of Theorem 3.1.

Of course, if the measurement is not performed, the property of the transport of particles is influenced by the potential and the interaction. It is well-known that when the potential is periodic the system shows the ballistic transport (the current is independent of the system size) and for random potentials the system shows the localization (Anderson localization). But our result shows that by performing the long time repeated measurements, the transport of the particles is described by the same stochastic process (SSEP) independent of the potential and the interaction. This fact concludes that the effect of measurement sometimes suppresses the transport (comparing to the ballistic case) and sometimes induces the transport (comparing to the localization case). 

Independence of the potential implies that even if the electric field is induced, the particles do not flow in the specific direction. Since the stochastic process is symmetric, some particles moves against the electric field. This means one can extract work from the system only by performing the measurement.

\section{Discussion and outlook}
In this paper, we considered the large time repeated measurements of configuration of particles, and showed that the classical stochastic process (SSEP) is obtained. From this stochastic process, diffusion equation emerges by the hydrodynamic limit. Although we dealt with only one-dimensional periodic lattice, our result is easily extended to any dimension and the lattice with boundary or infinite lattice. One of the key points of our result is that the effect of continuous measurement makes the way of particle diffusion universal. Our result suggests that in order to explain the universal behavior of diffusion in macroscopic world as seen in the diffusion equation from the quantum mechanical dynamics, disturbance from the environment would be necessary. But, projection measurement is the very strong disturbance. How can we obtain the diffusion equation from more physically natural dissipative quantum many body systems? This is our future work. 

In the present paper, we only use the fact that the measurement is described by 1-rank PVM to prove the former part of Theorem 3.1: $p^{\frac{1}{M}}_{\tau M}$ converges to a distribution $q_\tau$ which is described by the equation $\frac{d}{d\tau}q_\tau =Xq_\tau$. Thus, a part of our main result can be applied to general systems.
\vspace{6pt} 


\bibliography{stochastic}
\bibliographystyle{unsrt}

\end{document}